\documentclass[onecolumn]{revtex4-2}
\usepackage[
	bookmarks  = false,
	colorlinks = true,
	linkcolor  = blue,
	urlcolor   = blue,
	citecolor  = blue]{hyperref}
\usepackage{amsmath,amssymb,amsthm}
\usepackage{mathtools,braket}
\usepackage{graphicx,bm,bbm}
\usepackage[T1]{fontenc}
\usepackage{microtype}
\usepackage{orcidlink}
\newcommand{\appsection}[1]{\section{\MakeUppercase{#1}}}
\newtheorem{theorem}{Theorem}

\newtheorem{corollary}{Corollary}
\newtheorem{definition}{Definition}
\newtheorem{proposition}{Proposition}
\newtheorem{example}{Example}
\DeclareMathOperator{\Tr}{Tr}
\DeclareMathOperator{\PT}{PT}
\begin{document}
\title{$\bm{(2,m)}$-threshold quantum data hiding}
\author{Donghoon Ha\,\orcidlink{0000-0002-1698-9584}}
\affiliation{Department of Applied Mathematics and Institute of Natural Sciences, Kyung Hee University, Yongin 17104, Republic of Korea}
\author{Jeong San Kim\,\orcidlink{0000-0002-0420-8955}}
\email{freddie1@khu.ac.kr}
\affiliation{Department of Applied Mathematics and Institute of Natural Sciences, Kyung Hee University, Yongin 17104, Republic of Korea}
\begin{abstract}
Quantum data hiding aims to encode classical information into quantum states so that it can be recovered through suitable joint measurements while remaining inaccessible under local operations and classical communication (LOCC). Motivated by the question of whether quantum data hiding can be extended to a multiparty setting in which any pair of parties can recover the hidden information while it remains nearly inaccessible under LOCC, we investigate multiparty quantum state discrimination and propose a multiparty quantum data-hiding scheme for one classical bit shared among multiple parties. In the proposed scheme, any pair of parties can perfectly recover the hidden bit through a joint measurement, whereas LOCC measurements performed even by all parties can reveal only an arbitrarily small amount of information. We further derive bounds on the optimal LOCC discrimination of multiparty quantum states. The proposed scheme can be implemented using only separable states of low-dimensional quantum systems, thereby enhancing its practical feasibility.
\end{abstract}
\maketitle

\section{Introduction}\label{sec:int}
Data hiding is a communication protocol in which a specific piece of information is encoded into multiple shares and distributed among multiple users. 
No individual user can recover the hidden data from their own share alone, and recovering the hidden data requires a sufficient number of users to combine their shares. 
In classical data-hiding schemes, the hidden data is encoded into classical shares and can be reconstructed only through classical communication among the users. 
Therefore, classical data hiding fundamentally relies on restricting classical communication\cite{sham1979}.

Quantum data-hiding schemes provide a fundamentally stronger level of concealment than classical schemes\cite{terh2001,divi2002,egge2002}. 
In these schemes, classical data are encoded into quantum systems shared among multiple users. 
Unlike classical data-hiding schemes, the hidden data remains concealed when classical communication among the users is allowed. 
Moreover, the data remains hidden even under any protocol restricted to \emph{local operations and classical communication} (LOCC).
Recovering the hidden data instead requires global quantum operations enabled by quantum communication, shared entanglement, or direct interactions among the users.
Recently, related data-hiding phenomena have also been investigated in the resource theory of magic (nonstabilizerness) under restrictions to positive-Wigner-function measurements\cite{zhu2024,wang2026}.

The first quantum data-hiding scheme was proposed for concealing one classical bit between two parties\cite{terh2001}. 
Subsequently, a multiparty quantum data-hiding scheme was proposed, in which the hidden data remains concealed even if quantum communication is allowed among specified subsets of parties, whereas near-perfect recovery is possible for all other subsets\cite{egge2002}. 
Since then, quantum data hiding has been further developed in various directions\cite{lupo2016,lami2018,lami2021,ha20241}. 
More recently, an $(m,m)$-threshold quantum data-hiding scheme was proposed for concealing classical data of arbitrary size among $m$ parties\cite{ha20251}. 
The scheme requires the collaboration of all parties for recovery, representing the strongest possible collaboration requirement. 
It is therefore natural to ask whether quantum data hiding can be achieved under intermediate collaboration requirements.

Here, we propose a $(2,m)$-threshold quantum data-hiding scheme for concealing one classical bit among $m$ parties. The main contributions of this work are as follows. 
First, we construct a scheme in which any pair of parties can collaborate to perfectly recover the hidden bit through a joint measurement, whereas LOCC measurements performed even by all parties can reveal only an arbitrarily small amount of information. 
Second, we derive bounds on the optimal LOCC discrimination of multiparty quantum states, which are used to establish the concealment property of the proposed scheme. 
Third, we show that the scheme can be realized using only separable states of low-dimensional quantum systems, making its practical implementation more attainable.

\section{Multiparty quantum state discrimination}\label{sec:mqsd}
In this section, we introduce the mathematical framework for analyzing the discrimination of multiparty quantum states under \emph{positive-partial-transpose} (PPT) measurements. 
In particular, we first introduce the cones associated with positivity and the PPT condition and characterize their dual cones. 
These cone structures will then be used to formulate the optimal PPT discrimination probability as a convex optimization problem and to derive the dual characterizations required in the subsequent analysis.

For a multiparty Hilbert space $\mathcal{H}=\bigotimes_{k=1}^{m}\mathbb{C}^{d_{k}}$ with $m\geqslant2$ parties $A_{1},\ldots,A_{m}$ and local dimensions $d_{1},\ldots,d_{m}\geqslant2$, let $\mathbb{H}$ denote the set of all Hermitian operators acting on $\mathcal{H}$.
We further denote by $\mathbb{H}_{+}$ the set of all positive-semidefinite operators in $\mathbb{H}$, that is,
\begin{equation}\label{eq:hpdf}
\mathbb{H}_{+}=\{E\in\mathbb{H}\mid \bra{v}E\ket{v}\geqslant0~\mbox{for all}~\ket{v}\in\mathcal{H}\}.
\end{equation}
This cone structure of positive-semidefinite operators is useful for constructing the dual formulation of the PPT discrimination problem.
Note that $\mathbb{H}_{+}$ is a closed convex cone with self-duality,
\begin{equation}\label{eq:sdch}
\mathbb{H}_{+}^{*}=\mathbb{H}_{+},
\end{equation}
where the superscript $*$ denotes the dual cone of a given cone\cite{boyd2004,dual}.

A multiparty quantum state is described by a density operator $\rho$, that is, $\rho\in\mathbb{H}_{+}$ and $\Tr\rho=1$.
A measurement is represented by a positive operator-valued measure $\{M_{i}\}_{i}$, that is, $\{M_{i}\}_{i}\subseteq\mathbb{H}_{+}$ and $\sum_{i}M_{i}=\mathbbm{1}$ where $\mathbbm{1}$ denotes the identity operator in $\mathbb{H}$.

Now, we recall the concept of PPT, which will play a crucial role in providing a tractable relaxation of the optimization problems arising in quantum data hiding.
\begin{definition}\label{def:ppt}
For each $k\in\{1,\ldots,m\}$, we say that $E\in\mathbb{H}$ is \emph{$k$-PPT} if it is PPT with respect to party $A_k$, that is,
\begin{equation}\label{eq:kppt}
E^{\PT_{k}}\in\mathbb{H}_{+}
\end{equation}
where the superscript $\PT_{k}$ denotes the partial transposition of $E$ with respect to the party $A_{k}$ in the standard basis $\{\ket{i}\}_{i=0}^{d_{k}-1}$.
When $m=2$, we simply say that $E\in\mathbb{H}$ is \emph{PPT} if it is $1$-PPT (equivalently, $2$-PPT)\cite{pere1996,pptp}.
\end{definition}

For each $k\in\{1,\ldots,m\}$, we use $\mathbb{PPT}_{k}$ to denote the set of all $k$-PPT operators in $\mathbb{H}$, that is,
\begin{equation}\label{eq:pptk}
\mathbb{PPT}_{k}=\{E\in\mathbb{H}\mid E^{\PT_{k}}\in\mathbb{H}_{+}\}.
\end{equation}
Note that $\mathbb{PPT}_{k}$ is a closed convex cone with self-duality,
\begin{equation}\label{eq:sdpk}
\mathbb{PPT}_{k}^{*}=\mathbb{PPT}_{k},
\end{equation}
because $\mathbb{H}_{+}$ is closed convex and self dual, and $\Tr(EF)=\Tr(E^{\PT_{k}}F^{\PT_{k}})$ for all $E,F\in\mathbb{H}$.

Let us consider positive-semidefinite operators that are also PPT with respect to every party. The set of all such operators form a cone 
\begin{equation}\label{eq:pptp}
\mathbb{PPT}_{+}=\mathbb{H}_{+}\cap\mathbb{PPT}_{1}\cap\cdots\cap\mathbb{PPT}_{m},
\end{equation}
which is the intersection of the closed convex cones $\mathbb{H}_{+},\mathbb{PPT}_{1},\ldots,\mathbb{PPT}_{m}$.
The dual cone of $\mathbb{PPT}_{+}$ plays a central role in the dual formulation of optimal PPT discrimination derived later. 
It can be written as the following equivalent forms;
\begin{eqnarray}\label{eq:rwpt}
\mathbb{PPT}_{+}^{*}
&=&\left\{E_{0}+E_{1}+\cdots+E_{m}\,\middle|\,E_{0}\in\mathbb{H}_{+}^{*},~E_{1}\in\mathbb{PPT}_{1}^{*},\,\ldots,\,E_{m}\in\mathbb{PPT}_{m}^{*}\right\}\nonumber\\
&=&\left\{E_{0}+E_{1}+\cdots+E_{m}\,\middle|\,E_{0}\in\mathbb{H}_{+},~E_{1}\in\mathbb{PPT}_{1},\,\ldots,\,E_{m}\in\mathbb{PPT}_{m}\right\}\nonumber\\
&=&\left\{E_{0}+E_{1}^{\PT_{1}}+\cdots+E_{m}^{\PT_{m}}\,\middle|\,E_{0},E_{1},\ldots,E_{m}\in\mathbb{H}_{+}\right\},
\end{eqnarray}
where the second equality follows from Eqs.~\eqref{eq:sdch} and \eqref{eq:sdpk}, and the last equality follows from Eq.~\eqref{eq:pptk}.
The first equality in Eq.~\eqref{eq:rwpt} follows from
\begin{equation}\label{eq:dcif}
\textstyle
\left(\bigcap_{i}C_{i}\right)^{*}=\left\{\,\sum_{i}e_{i}\,\middle|\,e_{i}\in C_{i}^{*}~\mbox{for all}~i\,\right\}
\end{equation}
for any finite collection of closed convex cones $\{C_{i}\}_{i}$ in a real vector space\cite{sand1954}.

Now, let us consider the situation of discriminating multiparty quantum states $\rho_{0},\ldots,\rho_{n-1}$ from the ensemble
\begin{equation}\label{eq:menb}
\mathcal{E}=\{\eta_{i},\rho_{i}\}_{i=0}^{n-1},
\end{equation}
where the state $\rho_{i}$ is prepared with the probability $\eta_{i}$ for each $i\in\{0,\ldots,n-1\}$.
To guess the prepared state from $\mathcal{E}$, we use a measurement 
\begin{equation}\label{eq:wumt}
\mathcal{M}=\{M_{i}\}_{i=0}^{n-1},
\end{equation}
where the measurement outcome corresponding to $M_{i}$ leads to the decision that the prepared state is $\rho_{i}$, for each $i\in\{0,\ldots,n-1\}$.
The \emph{minimum-error discrimination} of $\mathcal{E}$ is to achieve the maximum average probability of correctly guessing the prepared state from the ensemble $\mathcal{E}$, that is,
\begin{equation}\label{eq:dpge}
p_{\sf G}(\mathcal{E})=\max_{\mathcal{M}}\sum_{i=0}^{n-1}\eta_{i}\Tr(\rho_{i}M_{i}),
\end{equation}
where the maximum is taken over all possible measurements\cite{hels1969}.

A measurement is called an \emph{LOCC measurement} if it can be realized by local operations on each of the parties $A_{1},\ldots,A_{m}$ together with classical communication among them\cite{benn1999,chit2014}.
When the available measurements are restricted to LOCC measurements, we denote by $p_{\sf L}(\mathcal{E})$ the maximum average probability of correctly guessing the prepared state from $\mathcal{E}$ in Eq.~\eqref{eq:menb}, that is,
\begin{equation}\label{eq:dple}
p_{\sf L}(\mathcal{E})=\max_{\textsf{LOCC}\,\mathcal{M}}\sum_{i=0}^{n-1}\eta_{i}\Tr(\rho_{i}M_{i}).
\end{equation}

PPT operations form a broader class containing LOCC operations and provide a useful convex-cone framework for analyzing restricted quantum measurements.
We say that a measurement $\{M_{i}\}_{i}$ is a \emph{PPT measurement} if $\{M_{i}\}_{i}\subseteq\mathbb{PPT}_{+}$\cite{virm2003,matt2009}.
Similarly, we define the optimal PPT discrimination probability of $\mathcal{E}$ by 
\begin{equation}\label{eq:dpte}
p_{\sf PPT}(\mathcal{E})=\max_{\textsf{PPT}\,\mathcal{M}}\sum_{i=0}^{n-1}\eta_{i}\Tr(\rho_{i}M_{i}),
\end{equation}
where the maximum is taken over all possible PPT measurements.
Since simply guessing the most probable state is an LOCC measurement, and every LOCC measurement is a PPT measurement\cite{chit2014,ha20251}, it follows from the definitions of $p_{\sf G}(\mathcal{E})$, $p_{\sf L}(\mathcal{E})$ and $p_{\sf PPT}(\mathcal{E})$ that
\begin{equation}\label{eq:ineq}
\tfrac{1}{n}\leqslant
\max\{\eta_{0},\ldots,\eta_{n-1}\}\leqslant
p_{\sf L}(\mathcal{E})\leqslant
p_{\sf PPT}(\mathcal{E})\leqslant
p_{\sf G}(\mathcal{E}).
\end{equation}

The cone characterization of $\mathbb{PPT}_{+}^{*}$ in Eq.~\eqref{eq:rwpt} can allow us to express the optimization over 
$\mathbb{PPT}_{+}$ in Eq.~\eqref{eq:pptp} as a dual form.
The following theorem provides this dual characterization for the PPT discrimination problem in Eq.~\eqref{eq:dpte}, which will serve as the main tool for deriving upper bounds on PPT, and hence LOCC, discrimination probabilities in the subsequent sections.
The proof of Theorem~\ref{thm:ubpt} is given in Appendix~\ref{app:ubpt}.
\begin{theorem}\label{thm:ubpt}
For a multiparty quantum state ensemble $\mathcal{E}=\{\eta_{i},\rho_{i}\}_{i=0}^{n-1}$ and 
\begin{equation}\label{eq:dhpt}
\mathbb{H}_{\sf PPT}(\mathcal{E}):=\{H\in\mathbb{H}\,|\, 
H-\eta_{i}\rho_{i}\in\mathbb{PPT}_{+}^{*}~\mbox{for all}~i=0,\ldots,n-1\,\},
\end{equation}
we have
\begin{equation}\label{eq:dppt}
p_{\sf PPT}(\mathcal{E})=\min_{H\in\mathbb{H}_{\sf PPT}(\mathcal{E})}\Tr H,
\end{equation}
where the minimum is taken over all possible $H\in\mathbb{H}_{\sf PPT}(\mathcal{E})$.
\end{theorem}

For the two-state ensembles used in our data-hiding construction, the dual characterization in Theorem~\ref{thm:ubpt} admits a particularly useful trace-norm formulation stated in the following theorem.
The proof of Theorem~\ref{thm:tspt} is given in Appendix~\ref{app:tspt}.
\begin{theorem}\label{thm:tspt}
For a two-state ensemble $\mathcal{E}=\{\eta_{0},\rho_{0};\eta_{1},\rho_{1}\}$ and 
\begin{equation}\label{eq:lame}
\Lambda_{\mathcal{E}}:=\eta_{0}\rho_{0}-\eta_{1}\rho_{1},
\end{equation}
we have
\begin{equation}\label{eq:tpte}
p_{\sf PPT}(\mathcal{E})=
\tfrac{1}{2}+\tfrac{1}{2}\min(\Tr|E_{0}|+\Tr|E_{1}|+\cdots+\Tr|E_{m}|)
\end{equation}
over all possible $\{E_{0},E_{1},\ldots,E_{m}\}\subseteq\mathbb{H}$ satisfying
\begin{equation}\label{eq:cdes}
E_{0}+E_{1}^{\PT_{1}}+\cdots+E_{m}^{\PT_{m}}=\Lambda_{\mathcal{E}},
\end{equation}
where $|E|$ denotes the positive square root of $E^{2}$ for $E\in\mathbb{H}$.
\end{theorem}

For the case of two parties ($m=2$), the identity $\mathbb{PPT}_{1}=\mathbb{PPT}_{2}$ implies that the cones $\mathbb{PPT}_{+}$ and $\mathbb{PPT}_{+}^{*}$ in Eqs.~\eqref{eq:pptp} and \eqref{eq:rwpt} can be expressed as
\begin{subequations}\label{eq:cmtp}
\begin{eqnarray}
\mathbb{PPT}_{+}&=&\mathbb{H}_{+}\cap\mathbb{PPT}_{1},\label{eq:mtpp}\\
\mathbb{PPT}_{+}^{*}&=&\left\{E_{0}+E_{1}^{\PT_{1}}\,\middle|\,E_{0},E_{1}\in\mathbb{H}_{+}\right\}.\label{eq:mtps}
\end{eqnarray}
\end{subequations}
Using Eq.~\eqref{eq:cmtp} and arguing as in the proof of Theorem~\ref{thm:tspt}, we obtain the following corollary.
\begin{corollary}\label{cor:tmtp}
For a two-party two-state ensemble $\mathcal{E}=\{\eta_{0},\rho_{0};\eta_{1},\rho_{1}\}$, we have
\begin{equation}\label{eq:tmtp}
p_{\sf PPT}(\mathcal{E})=
\tfrac{1}{2}+\tfrac{1}{2}\min(\Tr|E_{0}|+\Tr|E_{1}|),
\end{equation}
where the minimum is taken over all possible $\{E_{0},E_{1}\}\subseteq\mathbb{H}$ satisfying
\begin{equation}\label{eq:celm}
E_{0}+E_{1}^{\PT_{1}}=\Lambda_{\mathcal{E}}.
\end{equation}
\end{corollary}

\section{$(2,m)$-threshold scheme for quantum data hiding}\label{sec:mqdh}
In this section, we present a $(2,m)$-threshold data-hiding scheme for one classical bit shared among $m$ parties $A_{1},\ldots,A_{m}$. 
The hidden bit can be perfectly recovered by any pair of parties among $A_{1},\ldots,A_{m}$ through an appropriate joint measurement, while any LOCC measurement performed even by all $m$ parties reveals only an arbitrarily small amount of information.

To construct the proposed $(2,m)$-threshold data-hiding scheme, we consider a multiparty quantum system consisting of $m$ parties $A_{1},\ldots,A_{m}$, where each pair of parties shares a two-party subsystem; for each $k\in\{1,\ldots,m\}$, let the party $A_k$ consist of the $m-1$ subsystems
\begin{equation}
A_{k}^{(1)},\ldots,A_{k}^{(k-1)},A_{k}^{(k+1)},\ldots,A_{k}^{(m)}.
\end{equation}
For every pair $(k,k')$ with $1\leqslant k<k'\leqslant m$, the parties $A_k$ and $A_{k'}$ share the two-party subsystem $A_{k}^{(k')}A_{k'}^{(k)}$.
Figure~\ref{fig:exmt} illustrates the subsystems of each party when $m=3$.

\begin{figure}[!tt]
\centerline{\includegraphics{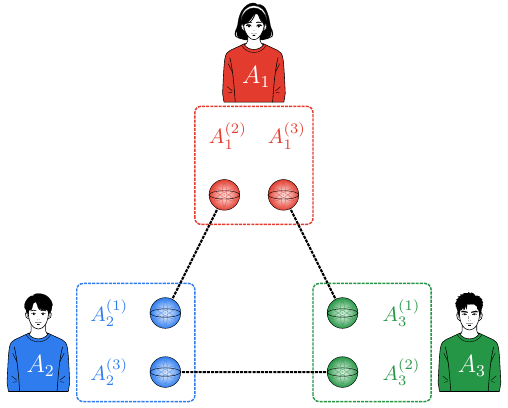}}
\caption{Subsystems of each party when $m=3$. 
The party $A_{1}$ consists of the subsystems $A_{1}^{(2)}$ and $A_{1}^{(3)}$, the party $A_{2}$ consists of the subsystems $A_{2}^{(1)}$ and $A_{2}^{(3)}$, and the party $A_{3}$ consists of the subsystems $A_{3}^{(1)}$ and $A_{3}^{(2)}$.
Each pair of parties shares a two-party subsystem: $A_{1}^{(2)}A_{2}^{(1)}$ between $A_{1}$ and $A_{2}$, $A_{1}^{(3)}A_{3}^{(1)}$ between $A_{1}$ and $A_{3}$, and $A_{2}^{(3)}A_{3}^{(2)}$ between $A_{2}$ and $A_{3}$.
}\label{fig:exmt}
\end{figure}

The following theorem shows that the optimal PPT discrimination probability of a multiparty quantum state ensemble can be bounded in terms of those of the underlying two-party state ensembles.
The proof of Theorem~\ref{thm:umpt} is given in Appendix~\ref{app:umpt}.
\begin{theorem}\label{thm:umpt}
When $\sigma_{0}^{(k,k')}$ and $\sigma_{1}^{(k,k')}$ are PPT states on the two-party system $A_{k}^{(k')}A_{k'}^{(k)}$ for all $k,k'\in\{1,\ldots,m\}$ with $k<k'$, the multiparty quantum state ensemble $\mathcal{E}=\{\eta_{0},\rho_{0};\eta_{1},\rho_{1}\}$ with
\begin{equation}\label{eq:ttps}
\eta_{0}=\eta_{1}=\frac{1}{2},~
\rho_{0}=\bigotimes_{\substack{k,k'=1\\ k<k'}}^{m}\sigma_{0}^{(k,k')},~
\rho_{1}=\bigotimes_{\substack{k,k'=1\\ k<k'}}^{m}\sigma_{1}^{(k,k')}
\end{equation}
satisfies
\begin{equation}\label{eq:sppt}
p_{\sf PPT}(\mathcal{E})\leqslant
\frac{1}{2}+\sum_{\substack{k,k'=1\\ k<k'}}^{m}\left(p_{\sf PPT}(\mathcal{E}^{(k,k')})-\frac{1}{2}\right)
\end{equation}
where each $\mathcal{E}^{(k,k')}$ is the two-party state ensemble defined by $\sigma_{0}^{(k,k')}$ and $\sigma_{1}^{(k,k')}$ with equal probabilities, that is,
\begin{equation}\label{eq:senb}
\mathcal{E}^{(k,k')}=\{\tfrac{1}{2},\sigma_{0}^{(k,k')};\tfrac{1}{2},\sigma_{1}^{(k,k')}\},
\end{equation}
for $k,k'\in\{1,\ldots,m\}$ with $k<k'$.
\end{theorem}

For an arbitrarily small $\epsilon>0$, there exist two-party orthogonal PPT states $\tau_{0}$ and $\tau_{1}$ such that
\begin{equation}\label{eq:teos}
p_{\sf PPT}(\mathcal{T})\leqslant\tfrac{1}{2}+\epsilon
\end{equation}
where $\mathcal{T}$ is the equiprobable state ensemble consisting of $\tau_{0}$ and $\tau_{1}$\cite{ha20252,mele2025,ha20253}, that is,
\begin{equation}\label{eq:enbt}
\mathcal{T}=\{\tfrac{1}{2},\tau_{0};\tfrac{1}{2},\tau_{1}\}.
\end{equation}
A construction of such an ensemble $\mathcal{T}$ is presented in Appendix~\ref{app:csse}. Specifically, starting from an orthogonal PPT two-state ensemble $\mathcal{Q}$, the construction considers the ensemble $\mathcal{Q}^{(L)}$ obtained by grouping the tensor-product states of $L$ independently prepared copies according to their parity. The two states in $\mathcal{Q}^{(L)}$ remain orthogonal and PPT, while their optimal PPT discrimination probability approaches the random-guessing probability $1/2$ exponentially as $L$ increases. Consequently, for any $\epsilon>0$, one can choose a sufficiently large $L$ and take $\mathcal{T}=\mathcal{Q}^{(L)}$ so that Eq.~\eqref{eq:teos} holds.
In Appendix~\ref{app:csse}, we provide an explicit construction of such ensemble $\mathcal{T}$ in Eq.~\eqref{eq:enbt} and establishes quantitatively that its PPT distinguishability can be made arbitrarily close to random guessing.

Using the ensemble $\mathcal{T}$ in Eq.~\eqref{eq:enbt} as a building block, we now describe a $(2,m)$-threshold data-hiding scheme for one classical bit.
For every pair $(k,k')$ with $1\leqslant k<k'\leqslant m$, the hider independently prepares a two-party quantum state $\tau_{b_{k,k'}}$ from the ensemble $\mathcal{T}$, where $b_{k,k'}\in\{0,1\}$ is chosen uniformly at random, and distributes it to the parties $A_{k}$ and $A_{k'}$.
This results in the multiparty quantum state 
\begin{equation}\label{eq:dhps}
\tau_{\vec{b}}=\bigotimes_{\substack{k,k'=1\\ k<k'}}^{m}[\tau_{b_{k,k'}}]_{A_{k}^{(k')}A_{k'}^{(k)}}
\end{equation}
where
\begin{equation}\label{eq:vecb}
\vec{b}=(b_{k,k'})_{1\leqslant k<k'\leqslant m}\in\{0,1\}^{\binom{m}{2}},
\end{equation}
which is prepared with probability $1/2^{\binom{m}{2}}$.
To conceal a bit $x\in\{0,1\}$, the hider broadcasts to the parties $A_{1},A_{2},\ldots,A_{m}$ the bit string 
\begin{equation}\label{eq:bsvc}
\vec{c}=(c_{k,k'})_{1\leqslant k<k'\leqslant m}\in\{0,1\}^{\binom{m}{2}}
\end{equation}
defined by
\begin{equation}\label{eq:vecc}
c_{k,k'}=
\left\{
\begin{array}{lcc}
b_{\alpha,\alpha'}\oplus x&,&(k,k')=(\alpha,\alpha'),\\
b_{\alpha,\alpha'}\oplus b_{k,k'}&,&(k,k')\neq(\alpha,\alpha'),
\end{array}
\right.
\end{equation}
where $\oplus$ is the modulo-$2$ addition and $(\alpha,\alpha')$ is a fixed but arbitrary pair satisfying $1\leqslant\alpha<\alpha'\leqslant m$.
Figure~\ref{fig:sch} illustrates the proposed $(2,m)$-threshold data-hiding scheme for the case of $m=3$ with $(\alpha,\alpha')=(1,2)$.

\begin{figure}[!tt]
\centerline{\includegraphics{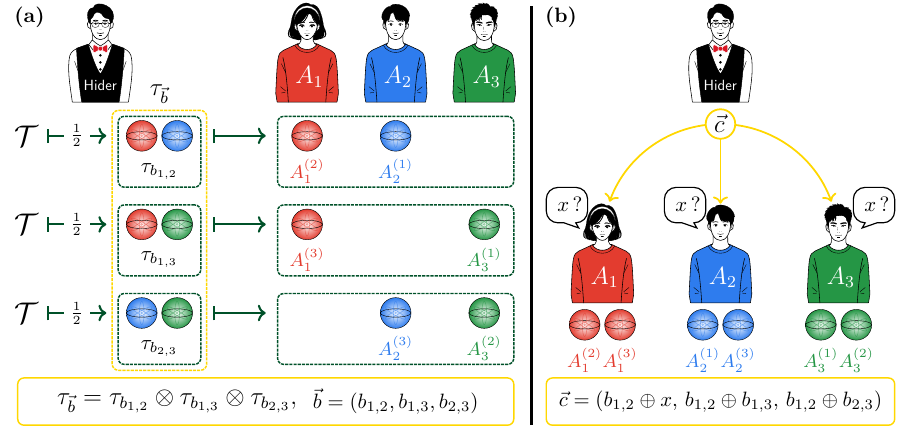}}
\caption{$(2,m)$-threshold data-hiding scheme for the case $m=3$ with $(\alpha,\alpha')=(1,2)$. 
(a) The hider independently prepares the two-party quantum states $\tau_{b_{1,2}}$, $\tau_{b_{1,3}}$, and $\tau_{b_{2,3}}$ from the ensemble $\mathcal{T}$, where each bit $b_{1,2}, b_{1,3}, b_{2,3}\in\{0,1\}$ is chosen uniformly at random, and distributes the states to the corresponding pairs of parties $(A_{1},A_{2})$, $(A_{1},A_{3})$, and $(A_{2},A_{3})$, respectively.
(b) To hide a bit $x\in\{0,1\}$, the hider broadcasts the bit string $\vec{c}=(c_{1,2},c_{1,3},c_{2,3})$ to the parties $A_{1}$, $A_{2}$, and $A_{3}$, where $c_{1,2}=b_{1,2}\oplus x$, $c_{1,3}=b_{1,2}\oplus b_{1,3}$ and $c_{2,3}=b_{1,2}\oplus b_{2,3}$.
}\label{fig:sch}
\end{figure}

As the bit $c_{\alpha,\alpha'}$ is publicly revealed, determining the hidden bit $x$ is equivalent to determining the bit $b_{\alpha,\alpha'}$ for the prepared state $\tau_{b_{\alpha,\alpha'}}$ on the system $A_{\alpha}^{(\alpha')}A_{\alpha'}^{(\alpha)}$.
As all bits in $\vec{c}$ except $c_{\alpha,\alpha'}$ are publicly revealed, determining $b_{\alpha,\alpha'}$ is equivalent to discriminate the states from the multiparty quantum state ensemble $\mathcal{E}=\{\eta_{0},\rho_{0};\,\eta_{1},\rho_{1}\}$ satisfying Eq.~\eqref{eq:ttps} and
\begin{align}\label{eq:epse}
\eta_{0}=\frac{1}{2},~&\sigma_{0}^{(k,k')}=\left\{
\begin{array}{ccc}
\tau_{0}&,&(k,k')=(\alpha,\alpha'),\\
\tau_{c_{k,k'}\oplus 0}&,&(k,k')\neq(\alpha,\alpha'),
\end{array}
\right.\nonumber\\
\eta_{1}=\frac{1}{2},~&\sigma_{1}^{(k,k')}=\left\{
\begin{array}{ccc}
\tau_{1}&,&(k,k')=(\alpha,\alpha'),\\
\tau_{c_{k,k'}\oplus 1}&,&(k,k')\neq(\alpha,\alpha'),
\end{array}
\right.
\end{align}
for any pair $(k,k')$ with $1\leqslant k<k'\leqslant m$.
Therefore, the maximum average probability of correctly guessing $x$ using only LOCC measurements becomes $p_{\sf L}(\mathcal{E})$.

From Theorem~\ref{thm:umpt} and Inequality~\eqref{eq:teos}, together with the construction of $\mathcal{E}$ in Eqs.~\eqref{eq:ttps} and \eqref{eq:epse}, we have
\begin{equation}\label{eq:inpl}
\frac{1}{2}\leqslant
p_{\sf L}(\mathcal{E})\leqslant
p_{\sf PPT}(\mathcal{E})\leqslant
\frac{1}{2}+\binom{m}{2}\left(p_{\sf PPT}(\mathcal{T})-\frac{1}{2}\right).
\end{equation}
Since $\epsilon$ can be chosen arbitrarily small, it follows from Inequality~\eqref{eq:inpl} that
$p_{\sf L}(\mathcal{E})$ can be made arbitrarily close to $\frac{1}{2}$.
In other words, $\epsilon$ can be chosen so small that any LOCC measurement reveals only an arbitrarily small amount of information about the hidden bit $x$.

For example, let $\mathcal{Q}$ denote the orthogonal separable two-state ensemble introduced in Example~\ref{ex:osse} of Appendix~\ref{app:csse}, and take $\mathcal{T}=\mathcal{Q}^{(L)}$, where $L$ denotes the number of independently prepared copies of $\mathcal{Q}$ used in the construction of the parity-encoded ensemble $\mathcal{Q}^{(L)}$. Then, as shown in Appendix~\ref{app:csse}, the upper bound in Inequality~\eqref{eq:inpl} can be explicitly obtained as
\begin{equation}\label{eq:inex}
\frac{1}{2}+\binom{m}{2}\left(p_{\sf PPT}(\mathcal{T})-\frac{1}{2}\right)=\frac{1}{2}+\binom{m}{2}\left(p_{\sf PPT}(\mathcal{Q}^{(L)})-\frac{1}{2}\right)\leqslant\frac{1}{2}+\frac{1}{2}\binom{m}{2}\left(\frac{35}{36}\right)^{\frac{L}{2}}.
\end{equation}
This particular choice of $\mathcal{T}$ therefore provides an explicit upper bound showing that $p_{\sf PPT}(\mathcal{E})$ converges exponentially to $1/2$ as the number $L$ of copies increases.

Furthermore, it follows from the construction of $\mathcal{E}$ in Eqs.~\eqref{eq:ttps} and \eqref{eq:epse}, together with the orthogonality of $\tau_{0}$ and $\tau_{1}$, that the states $\rho_{0}$ and $\rho_{1}$ can be perfectly discriminated by performing an appropriate measurement on any two-party subsystem $A_{k}^{(k')}A_{k'}^{(k)}$ with $1\leqslant k<k'\leqslant m$.
In other words, any two parties can perfectly recover the hidden bit $x$ by jointly measuring their shared subsystem.
Thus, the ensemble $\mathcal{T}$ in Eq.~\eqref{eq:enbt} can be used to construct a $(2,m)$-threshold data-hiding scheme that conceals one classical bit.

\section{Conclusion and Discussion}\label{sec:disc}
We have considered the discrimination of multiparty quantum states and established a multiparty quantum data-hiding scheme for one classical bit to be shared among multiple parties. 
In the proposed scheme, the hidden bit can be perfectly recovered by any pair of parties through a joint measurement, whereas LOCC measurements performed even by all parties reveal only an arbitrarily small amount of information. 
Moreover, we have established upper bounds on the optimal success probability of LOCC discrimination of multiparty quantum states.

As shown in Appendix~\ref{app:csse}, the proposed data-hiding scheme can be implemented using only separable states of low-dimensional quantum systems.
Although the explicit construction considered here is based on a two-qutrit state ensemble (Example~\ref{ex:osse}), analogous constructions are possible using a two-qu$d$it state ensemble for arbitrary local dimension $d\geqslant2$\cite{ha20252,mele2025,ha20253}.
The key ingredient is a two-party orthogonal separable state ensemble constructed by quantum sequences whose optimal PPT discrimination converges exponentially to random guessing.
This not only demonstrates that entanglement is not a necessary resource for realizing multiparty quantum data-hiding schemes, but also improves the practical feasibility of their implementation.

As the present scheme hides a single classical bit, it is natural to ask whether the construction can be extended to encode multiple bits while preserving the same threshold property. 
It would also be interesting to investigate $(k,m)$-threshold quantum data-hiding schemes in which the threshold $k$ is larger than $2$. 
Such schemes would bridge the gap between the existing $(m,m)$-threshold schemes\cite{ha20251} and the $(2,m)$-threshold scheme developed in this work.

Another interesting direction is to investigate whether similar multiparty data-hiding phenomena can arise under other physically motivated restrictions on measurements. In particular, recent studies have demonstrated data-hiding and indistinguishability phenomena under positive-Wigner-function measurements in the resource theory of magic\cite{zhu2024,wang2026}. Although these measurements are conceptually different from LOCC and PPT measurements, both settings illustrate how restrictions on accessible measurements can give rise to information-hiding effects. Extending the present multiparty framework to such measurement restrictions may therefore provide a broader understanding of quantum data hiding from a resource-theoretic perspective.

\section*{Acknowledgments}
This work was supported by Korea Research Institute for defense Technology planning and advancement (KRIT) grant funded by Defense Acquisition Program Administration(DAPA)(KRIT-CT-23–031) and the Institute for Information \& Communications Technology Planning \& Evaluation(IITP) grant funded by the Korean government(MSIP)(Grant No. RS-2025-02304540). JSK was supported by Creation of the Quantum Information Science R\&D Ecosystem(Grant No. 2022M3H3A106307411) through the National Research Foundation of Korea(NRF) funded by the Korean government(Ministry of Science and ICT).

\appendix
\appsection{Proof of Theorem~\ref{thm:ubpt}}\label{app:ubpt}
In this Appendix, we prove Theorem~\ref{thm:ubpt} by showing that
\begin{subequations}\label{eq:plgq}
\begin{eqnarray}
p_{\sf PPT}(\mathcal{E})&\leqslant&\min_{H\in\mathbb{H}_{\sf PPT}(\mathcal{E})}\Tr H,\label{eq:pleq}\\
p_{\sf PPT}(\mathcal{E})&\geqslant&\min_{H\in\mathbb{H}_{\sf PPT}(\mathcal{E})}\Tr H.\label{eq:pgeq}
\end{eqnarray}
\end{subequations}

\begin{proof}[Proof of Inequality~\eqref{eq:pleq}]
Let us assume that a PPT measurement $\mathcal{M}=\{M_{i}\}_{i=0}^{n-1}$ provides $p_{\sf PPT}(\mathcal{E})$ in Eq.~\eqref{eq:dpte}.
For any $H\in\mathbb{H}_{\sf PPT}(\mathcal{E})$, we have
\begin{equation}\label{eq:pupb}
p_{\sf PPT}(\mathcal{E})=\sum_{i=0}^{n-1}\eta_{i}\Tr(\rho_{i}M_{i})
\leqslant\sum_{i=0}^{n-1}\eta_{i}\Tr(\rho_{i}M_{i})+\sum_{i=0}^{n-1}\Tr[M_{i}(H-\eta_{i}\rho_{i})]
=\Tr H,
\end{equation}
where the first equality is from the assumption of $\mathcal{M}$, the inequality is due to $M_{i}\in\mathbb{PPT}_{+}$ and $H-\eta_{i}\rho_{i}\in\mathbb{PPT}_{+}^{*}$ for all $i\in\{0,\ldots,n-1\}$, and the last equality is by $\sum_{i=0}^{n-1}M_{i}=\mathbbm{1}$.
Thus, Inequality~\eqref{eq:pupb} leads us to Inequality~\eqref{eq:pleq}.
\qedhere
\end{proof}

\begin{proof}[Proof of Inequality~\eqref{eq:pgeq}]
Let us consider the set
\begin{equation}\label{eq:cset}
\mathcal{S}(\mathcal{E})=\Bigg\{\Bigg(\sum_{i=0}^{n-1}\eta_{i}\Tr(\rho_{i}M_{i})-p,\mathbbm{1}-\sum_{i=0}^{n-1}M_{i}\Bigg)\in\mathbb{R}\times\mathbb{H}\,\Bigg|
\,p\in\mathbb{R}~\mbox{with}~p>p_{\sf PPT}(\mathcal{E}),~M_{0},\ldots,M_{n-1}\in\mathbb{PPT}_{+}\Bigg\}
\end{equation}
where $\mathbb{R}$ is the set of all real numbers.
By the convexity of $\mathbb{PPT}_{+}$ in Eq.~\eqref{eq:pptp}, the set $\mathcal{S}(\mathcal{E})$ is convex. 
Moreover, $\mathcal{S}(\mathcal{E})$ does not contain the origin $(0,\mathbb{O})$ of $\mathbb{R}\times\mathbb{H}$; otherwise, there exists a PPT measurement $\{M_{i}\}_{i=0}^{n-1}$ satisfying
\begin{equation}\label{eq:rctd}
\sum_{i=0}^{n-1}\eta_{i}\Tr(\rho_{i}M_{i})>p_{\sf PPT}(\mathcal{E}),
\end{equation}
which contradicts the optimality of $p_{\sf PPT}(\mathcal{E})$ in Eq.~\eqref{eq:dpte}.
Here, $\mathbb{O}$ denotes the zero operator in $\mathbb{H}$. 
We also note that the Cartesian product $\mathbb{R}\times\mathbb{H}$ is a real inner-product space equipped with the inner product
\begin{equation}\label{eq:inpd}
\langle (a,A),(b,B)\rangle=ab+\Tr(AB)
\end{equation}
for all $(a,A),(b,B)\in\mathbb{R}\times\mathbb{H}$.

Since $\mathcal{S}(\mathcal{E})$ and the single-element set $\{(0,\mathbb{O})\}$ are disjoint convex sets, it follows from the separating hyperplane theorem\cite{boyd2004,sht} that there exists 
\begin{equation}\label{eq:ggeq}
(\gamma,\Gamma)\in(\mathbb{R}\times\mathbb{H})\setminus\{(0,\mathbb{O})\}
\end{equation}
satisfying
\begin{equation}\label{eq:shtc}
\langle(a,A),(\gamma,\Gamma)\rangle\leqslant0
\end{equation}
for all $(a,A)\in\mathcal{S}(\mathcal{E})$.

Assume that
\begin{subequations}\label{eq:scot}
\begin{gather}
\Tr\Gamma\leqslant\gamma p_{\sf PPT}(\mathcal{E}),\label{eq:fsco}\\
\{\Gamma-\gamma\eta_{i}\rho_{i}\}_{i=0}^{n-1}\subseteq\mathbb{PPT}_{+}^{*},\label{eq:ssco}\\ 
\gamma>0.\label{eq:tsco}
\end{gather}
\end{subequations}
This assumption implies Inequality~\eqref{eq:pgeq} because
\begin{equation}\label{eq:thtg}
\min_{H\in\mathbb{H}_{\sf PPT}(\mathcal{E})}\Tr H\leqslant \Tr(\Gamma/\gamma)\leqslant p_{\sf PPT}(\mathcal{E}),
\end{equation}
where the first inequality is from Inclusion~\eqref{eq:ssco} and Inequality~\eqref{eq:tsco} together with the definition of $\mathbb{H}_{\sf PPT}(\mathcal{E})$ in Eq.~\eqref{eq:dhpt}, and the second inequality is due to Inequalities~\eqref{eq:fsco} and \eqref{eq:tsco}.
It remains to establish Condition~\eqref{eq:scot}.
\qedhere
\end{proof}

\begin{proof}[Proof of \eqref{eq:fsco}]
From the inner product in Eq.~\eqref{eq:inpd}, Inequality~\eqref{eq:shtc} can be rewritten as
\begin{equation}\label{eq:tgmr}
\Tr\Gamma-\sum_{i=0}^{n-1}\Tr[M_{i}(\Gamma-\gamma\eta_{i}\rho_{i})]\leqslant\gamma p
\end{equation}
for all $p>p_{\sf PPT}(\mathcal{E})$ and all $\{M_{i}\}_{i=0}^{n-1}\subseteq\mathbb{PPT}_{+}$.
If $M_{i}=\mathbb{O}$ for all $i\in\{0,\ldots,n-1\}$, then Inequality~\eqref{eq:tgmr} reduces to
\begin{equation}\label{eq:tred}
\Tr\Gamma\leqslant\gamma p
\end{equation}
for all $p>p_{\sf PPT}(\mathcal{E})$.
Taking the limit of $p$ to $p_{\sf PPT}(\mathcal{E})$ gives Inequality~\eqref{eq:fsco}.
\qedhere
\end{proof}

\begin{proof}[Proof of \eqref{eq:ssco}]
For each $j\in\{0,\ldots,n-1\}$, let $M_{j}\in\mathbb{PPT}_{+}$ be arbitrary and set $M_{i}=\mathbb{O}$ for all $i\in\{0,\ldots,n-1\}\setminus\{j\}$. 
Applying Inequality~\eqref{eq:tgmr} and taking the limit of $p$ to $p_{\sf PPT}(\mathcal{E})$, we have
\begin{equation}\label{eq:mjpr}
\Tr\Gamma-\Tr[M_{j}(\Gamma-\gamma\eta_{j}\rho_{j})]\leqslant\gamma p_{\sf PPT}(\mathcal{E}).
\end{equation}

Assume that $\Gamma-\gamma\eta_{j}\rho_{j}\notin\mathbb{PPT}_{+}^{*}$. 
By the definition of the dual cone $\mathbb{PPT}_{+}^{*}$\cite{dual}, there exists $M\in\mathbb{PPT}_{+}$ such that
\begin{equation}\label{eq:dcpt}
\Tr[M(\Gamma-\gamma\eta_{j}\rho_{j})]<0.
\end{equation}
Since $\mathbb{PPT}_{+}$ is a cone, $tM\in\mathbb{PPT}_{+}$ for every $t>0$.
Therefore, $\{M_{i}\}_{i=0}^{n-1}$ defined by $M_{j}=tM$ for $t>0$ and $M_{i}=\mathbb{O}$ for all $i\in\{0,\ldots,n-1\}\setminus\{j\}$ satisfies $\{M_{i}\}_{i=0}^{n-1}\subseteq\mathbb{PPT}_{+}$.

Now, Inequality~\eqref{eq:mjpr} can be rewritten as
\begin{equation}\label{eq:aptm}
\Tr\Gamma -t\Tr[M(\Gamma-\gamma\eta_{j}\rho_{j})]\leqslant\gamma p_{\sf PPT}(\mathcal{E}).
\end{equation}
By Inequality~\eqref{eq:dcpt}, the left-hand side of Inequality~\eqref{eq:aptm} tends to $\infty$ as $t\to\infty$, whereas the right-hand side remains finite. 
This contradiction proves that $\Gamma-\gamma\eta_{j}\rho_{j}\in\mathbb{PPT}_{+}^{*}$.
Since the choice of $j\in\{0,\ldots,n-1\}$ can be arbitrary, Inclusion~\eqref{eq:ssco} is true.
\qedhere
\end{proof}

\begin{proof}[Proof of \eqref{eq:tsco}]
To show $\gamma\geqslant0$, suppose that $\gamma<0$. 
Since Inequality~\eqref{eq:tred} holds for all $p>p_{\sf PPT}(\mathcal{E})$, letting $p\to\infty$ causes its right-hand side to diverge to $-\infty$, whereas its left-hand side remains finite.
This contradiction implies that $\gamma\geqslant0$.

Now, assume that $\gamma=0$. 
In this case, Inclusion~\eqref{eq:ssco} and Inequality~\eqref{eq:tred} reduce to
\begin{equation}\label{eq:geqz}
\Gamma\in\mathbb{PPT}_{+}^{*},~\Tr\Gamma\leqslant0.
\end{equation}
Since Eq.~\eqref{eq:rwpt} implies $\Tr E>0$ for all $E\in\mathbb{PPT}_{+}^{*}$ with $E\neq\mathbb{O}$, it follows from Condition~\eqref{eq:geqz} that
\begin{equation}\label{eq:cgez}
\Gamma=\mathbb{O}.
\end{equation}
This contradicts Inclusion~\eqref{eq:ggeq}. 
Thus, Inequality~\eqref{eq:tsco} holds.
\qedhere
\end{proof}

\appsection{Proof of Theorem~\ref{thm:tspt}}\label{app:tspt}
In this Appendix, we prove Theorem~\ref{thm:tspt} by establishing 
\begin{subequations}\label{eq:plgw}
\begin{eqnarray}
p_{\sf PPT}(\mathcal{E})&\leqslant&\frac{1}{2}+\frac{1}{2}\min\sum_{k=0}^{m}\Tr|E_{k}|,\label{eq:plqw}\\
p_{\sf PPT}(\mathcal{E})&\geqslant&\frac{1}{2}+\frac{1}{2}\min\sum_{k=0}^{m}\Tr|E_{k}|,\label{eq:pgqw}
\end{eqnarray}
\end{subequations}
where the minimum is taken over all possible $\{E_{k}\}_{k=0}^{m}\subseteq\mathbb{H}$ satisfying Eq.~\eqref{eq:cdes}.

\begin{proof}[Proof of Inequality~\eqref{eq:plqw}]
For an arbitrary $\{E_{k}\}_{k=0}^{m}\subseteq\mathbb{H}$ with Eq.~\eqref{eq:cdes}, let us consider the Hermitian operator
\begin{equation}\label{eq:aech}
H=\frac{1}{2}\left[\eta_{0}\rho_{0}+\eta_{1}\rho_{1}+
E_{0}^{(+)}+E_{0}^{(-)}+E_{1}^{(+)\PT_{1}}+E_{1}^{(-)\PT_{1}}+\cdots+E_{m}^{(+)\PT_{m}}+E_{m}^{(-)\PT_{m}}\right]
\end{equation}
where $E_{k}^{(+)}$ and $E_{k}^{(-)}$ are the positive and negative parts of $E_k$, that is, positive-semidefinite operators with orthogonal supports satisfying
\begin{equation}\label{eq:eepm}
E_{k}^{(+)}-E_{k}^{(-)}=E_{k}
\end{equation}
for each $k\in\{0,1,\ldots,m\}$.
We note that Eq.~\eqref{eq:cdes} implies
\begin{eqnarray}
H-\eta_{0}\rho_{0}&=&E_{0}^{(-)}+E_{1}^{(-)\PT_{1}}+\cdots+E_{m}^{(-)\PT_{m}},\nonumber\\
H-\eta_{1}\rho_{1}&=&E_{0}^{(+)}+E_{1}^{(+)\PT_{1}}+\cdots+E_{m}^{(+)\PT_{m}}.\label{eq:hdpt}
\end{eqnarray}

Since $E_{k}^{(+)}$ and $E_{k}^{(-)}$ are positive semidefinite for all $k\in\{0,1,\ldots,m\}$, it follows from Eq.~\eqref{eq:rwpt} that
\begin{eqnarray}
E_{0}^{(-)}+E_{1}^{(-)\PT_{1}}+\cdots+E_{m}^{(-)\PT_{m}}&\in&\mathbb{PPT}_{+}^{*},\nonumber\\
E_{0}^{(+)}+E_{1}^{(+)\PT_{1}}+\cdots+E_{m}^{(+)\PT_{m}}&\in&\mathbb{PPT}_{+}^{*}.\label{eq:epms}
\end{eqnarray}
From Eqs.~\eqref{eq:hdpt} and \eqref{eq:epms}, we have
\begin{equation}\label{eq:hinh}
H\in\mathbb{H}_{\sf PPT}(\mathcal{E}),
\end{equation}
where $\mathbb{H}_{\sf PPT}(\mathcal{E})$ is defined in Eq.~\eqref{eq:dhpt}.
Therefore, we have
\begin{equation}\label{eq:pupp}
p_{\sf PPT}(\mathcal{E})\leqslant\Tr H
=\frac{1}{2}\Tr\left[\eta_{0}\rho_{0}+\eta_{1}\rho_{1}+\sum_{k=0}^{m}\left(E_{k}^{(+)}+E_{k}^{(-)}\right)\right]
=\frac{1}{2}+\frac{1}{2}\sum_{k=0}^{m}\Tr|E_{k}|,
\end{equation}
where the inequality follows from Theorem~\ref{thm:ubpt} and Eq.~\eqref{eq:hinh}, the first equality follows from the invariance of trace under partial transposition, and the second equality follows from $E_{k}^{(+)}+E_{k}^{(-)}=|E_{k}|$ for all $k\in\{0,1,\ldots,m\}$.
Since the choice of $\{E_{k}\}_{k=0}^{m}\subseteq\mathbb{H}$ satisfying Eq.~\eqref{eq:cdes} can be arbitrary, Inequality~\eqref{eq:plqw} is true.
\qedhere
\end{proof}

\begin{proof}[Proof of Inequality~\eqref{eq:pgqw}]
From Theorem~\ref{thm:ubpt}, there exists a Hermitian operator $H$ with
\begin{subequations}\label{eq:chop}
\begin{gather}
H\in\mathbb{H}_{\sf PPT}(\mathcal{E}),\label{eq:hiin}\\
p_{\sf PPT}(\mathcal{E})=\Tr H.\label{eq:peqh}
\end{gather}
\end{subequations}
Inclusion~\eqref{eq:hiin} and Eq.~\eqref{eq:rwpt} imply the existence of $\{E_{k}^{(0)},E_{k}^{(1)}\}_{k=0}^{m}\subseteq\mathbb{H}_{+}$ satisfying
\begin{eqnarray}
H-\eta_{0}\rho_{0}&=&E_{0}^{(0)}+E_{1}^{(0)\PT_{1}}+\cdots+E_{m}^{(0)\PT_{m}},\nonumber\\
H-\eta_{1}\rho_{1}&=&E_{0}^{(1)}+E_{1}^{(1)\PT_{1}}+\cdots+E_{m}^{(1)\PT_{m}}.\label{eq:hmee}
\end{eqnarray}

Now, let us define 
\begin{equation}\label{eq:ekek}
E_{k}=E_{k}^{(1)}-E_{k}^{(0)}
\end{equation}
for each $k\in\{0,1,\ldots,m\}$.
We have
\begin{eqnarray}
p_{\sf PPT}(\mathcal{E})=\Tr H
&=&\frac{1}{2}+\frac{1}{2}\Tr\left(2H-\eta_{0}\rho_{0}-\eta_{1}\rho_{1}\right)\nonumber\\
&=&\frac{1}{2}+\frac{1}{2}\sum_{k=0}^{m}\Tr\left(E_{k}^{(0)}+E_{k}^{(1)}\right)\nonumber\\
&\geqslant&\frac{1}{2}+\frac{1}{2}\sum_{k=0}^{m}\Tr\left|E_{k}^{(1)}-E_{k}^{(0)}\right|
=\frac{1}{2}+\frac{1}{2}\sum_{k=0}^{m}\Tr|E_{k}|,\label{eq:ttrh}
\end{eqnarray}
where the first equality follows from Eq.~\eqref{eq:peqh}, the third equality follows from Eq.~\eqref{eq:hmee} together with the invariance of trace under partial transposition, and the inequality follows from the triangle inequality of trace norm.
Moreover, subtracting the two identities in Eq.~\eqref{eq:hmee} shows that $\{E_{k}\}_{k=0}^{m}$ satisfies Eq.~\eqref{eq:cdes}.
Therefore, Inequality~\eqref{eq:ttrh} implies Inequality~\eqref{eq:pgqw}.
\qedhere
\end{proof}

\appsection{Proof of Theorem~\ref{thm:umpt}}\label{app:umpt}
For each pair $(k,k')$ with $1\leqslant k<k'\leqslant m$, Corollary~\ref{cor:tmtp} implies the existence of Hermitian operators $E_{0}^{(k,k')}$ and $E_{1}^{(k,k')}$ satisfying
\begin{subequations}\label{eq:pree}
\begin{gather}
E_{0}^{(k,k')}+E_{1}^{(k,k')\PT_{k}}=\tfrac{1}{2}\sigma_{0}^{(k,k')}-\tfrac{1}{2}\sigma_{1}^{(k,k')},\label{eq:ekkp}\\
p_{\sf PPT}(\mathcal{E}^{(k,k')})=\tfrac{1}{2}+\tfrac{1}{2}\big(\Tr|E_{0}^{(k,k')}|+\Tr|E_{1}^{(k,k')}|\big).\label{eq:kkpt}
\end{gather}
\end{subequations}
Let us define 
\begin{eqnarray}
E_{0}&=&\sum_{l=1}^{m-1}\sum_{l'=l+1}^{m}\sigma_{0}^{(1,2)}\otimes\cdots\otimes\sigma_{0}^{(l,l'-1)}\otimes E_{0}^{(l,l')}\otimes\sigma_{1}^{(l,l'+1)}\otimes\cdots\otimes\sigma_{1}^{(m-1,m)},\nonumber\\
E_{k}&=&\sum_{k'=k+1}^{m}\Big[\sigma_{0}^{(1,2)}\otimes\cdots\otimes\sigma_{0}^{(k,k'-1)}\otimes E_{1}^{(k,k')\PT_{k}}\otimes\sigma_{1}^{(k,k'+1)}\otimes\cdots\otimes\sigma_{1}^{(m-1,m)}\Big]^{\PT_{k}}\label{eq:soek}
\end{eqnarray}
for all $k\in\{1,\ldots,m\}$. 
For this $\{E_{k}\}_{k=0}^{m}$, Eq.~\eqref{eq:cdes} holds because
\begin{eqnarray}
\Lambda_{\mathcal{E}}
&=&\eta_{0}\rho_{0}-\eta_{1}\rho_{1}\nonumber\\
&=&\frac{1}{2}\bigotimes_{\substack{k,k'=1\\ k<k'}}^{m}\sigma_{0}^{(k,k')}-\frac{1}{2}\bigotimes_{\substack{k,k'=1\\ k<k'}}^{m}\sigma_{1}^{(k,k')}\nonumber\\
&=&\sum_{k=1}^{m-1}\sum_{k'=k+1}^{m}\sigma_{0}^{(1,2)}\otimes\cdots\otimes\sigma_{0}^{(k,k'-1)}\otimes(\tfrac{1}{2}\sigma_{0}^{(k,k')}-\tfrac{1}{2}\sigma_{1}^{(k,k')})\otimes\sigma_{1}^{(k,k'+1)}\otimes\cdots\otimes\sigma_{1}^{(m-1,m)}\nonumber\\
&=&\sum_{k=1}^{m-1}\sum_{k'=k+1}^{m}\sigma_{0}^{(1,2)}\otimes\cdots\otimes\sigma_{0}^{(k,k'-1)}\otimes(E_{0}^{(k,k')}+E_{1}^{(k,k')\PT_{k}})\otimes\sigma_{1}^{(k,k'+1)}\otimes\cdots\otimes\sigma_{1}^{(m-1,m)}\nonumber\\
&=&E_{0}+E_{1}^{\PT_{1}}+\cdots+E_{m}^{\PT_{m}},\label{eq:rsoc}
\end{eqnarray}
where the first and second equalities follow from Eqs.~\eqref{eq:lame} and \eqref{eq:ttps}, respectively, the fourth equality follows from Eq.~\eqref{eq:ekkp}, and the last equality follows from the definition of $\{E_{k}\}_{k=0}^{m}$ in Eq.~\eqref{eq:soek}.
We further note that the third equality in Eq.~\eqref{eq:rsoc} follows from the identity
\begin{equation}\label{eq:xytp}
\bigotimes_{l=1}^{L}X_{l}-\bigotimes_{l=1}^{L}Y_{l}=\sum_{l=1}^{L}X_{1}\otimes\cdots\otimes X_{l-1}\otimes(X_{l}-Y_{l})\otimes Y_{l+1}\otimes\cdots\otimes Y_{L}
\end{equation}
for any operators $X_{1},Y_{1},\ldots,X_{L},Y_{L}$.

For each $k\in\{1,\ldots,m\}$, we have
\begin{align}
\Tr|E_{k}|\leqslant\sum_{k'=k+1}^{m}&\Tr\Big|\Big[\sigma_{0}^{(1,2)}\otimes\cdots\otimes\sigma_{0}^{(k,k'-1)}\otimes E_{1}^{(k,k')\PT_{k}}\otimes\sigma_{1}^{(k,k'+1)}\otimes\cdots\otimes\sigma_{1}^{(m-1,m)}\Big]^{\PT_{k}}\Big|\nonumber\\
=\sum_{k'=k+1}^{m}&\Tr\big|\sigma_{0}^{(1,2)\PT_{k}}\big|\cdots\Tr\big|\sigma_{0}^{(k,k'-1)\PT_{k}}\big|\Tr\big|E_{1}^{(k,k')}\big|\Tr\big|\sigma_{1}^{(k,k'+1)\PT_{k}}\big|\cdots\Tr\big|\sigma_{1}^{(m-1,m)\PT_{k}}\big|\nonumber\\
=\sum_{k'=k+1}^{m}&\Tr|E_{1}^{(k,k')}|,\label{eq:upek}
\end{align}
where the inequality follows from the triangle inequality of trace norm, the first equality follows from $\Tr\big|\bigotimes_{l=1}^{L}X_{l}\big|=\prod_{l=1}^{L}\Tr|X_{l}|$ for any Hermitian operators $X_{1},\ldots,X_{L}$, and the second equality follows from the fact that $\sigma_{0}^{(l,l')}$ and $\sigma_{1}^{(l,l')}$ are PPT states for all $l,l'\in\{1,\ldots,m\}$ with $l<l'$.
Using the same argument as in the derivation of Eq.~\eqref{eq:upek}, we have
\begin{eqnarray}
\Tr|E_{0}|&\leqslant&\sum_{l=1}^{m-1}\sum_{l'=l+1}^{m}\Tr\left|\sigma_{0}^{(1,2)}\otimes\cdots\otimes\sigma_{0}^{(l,l'-1)}\otimes E_{0}^{(l,l')}\otimes\sigma_{1}^{(l,l'+1)}\otimes\cdots\otimes\sigma_{1}^{(m-1,m)}\right|\nonumber\\
&=&\sum_{l=1}^{m-1}\sum_{l'=l+1}^{m}\Tr\big|\sigma_{0}^{(1,2)}\big|\cdots\Tr\big|\sigma_{0}^{(l,l'-1)}\big|\Tr\big|E_{0}^{(l,l')}\big|\Tr\big|\sigma_{1}^{(l,l'+1)}\big|\cdots\Tr\big|\sigma_{1}^{(m-1,m)}\big|\nonumber\\
&=&\sum_{l=1}^{m-1}\sum_{l'=l+1}^{m}\Tr|E_{0}^{(l,l')}|.\label{eq:upez}
\end{eqnarray}
Thus, we have
\begin{eqnarray}
p_{\sf PPT}(\mathcal{E})&\leqslant&\frac{1}{2}+\frac{1}{2}\sum_{k=0}^{m}\Tr|E_{k}|\nonumber\\
&\leqslant&\frac{1}{2}+\frac{1}{2}\sum_{k=1}^{m-1}\sum_{k'=k+1}^{m}\big(\Tr|E_{0}^{(k,k')}|+\Tr|E_{1}^{(k,k')}|\big)\nonumber\\
&=&\frac{1}{2}+\sum_{k=1}^{m-1}\sum_{k'=k+1}^{m}\left(p_{\sf PPT}(\mathcal{E}^{(k,k')})-\frac{1}{2}\right),
\label{eq:mpup}
\end{eqnarray}
where the first inequality follows from Theorem~\ref{thm:tspt}, the second inequality follows from Inequalities~\eqref{eq:upek} and \eqref{eq:upez}, and the equality follows from Eq.~\eqref{eq:kkpt}.

\appsection{Construction of orthogonal PPT-state ensembles with near-random PPT discrimination}\label{app:csse}
In this Appendix, we present a method for constructing orthogonal PPT state ensembles whose optimal PPT discrimination probability is arbitrarily close to random guessing. 
To this end, we first recall the notion of a \emph{quantum sequence ensemble}\cite{ha20254,ha20261}.

For a two-party two-state ensemble $\mathcal{E}=\{\eta_{0},\rho_{0};\eta_{1},\rho_{1}\}$ and a positive integer $L$, let us consider the situation of discriminating quantum sequence, each consisting of $L$ quantum states independently prepared from $\mathcal{E}$.
The preparation of such a quantum sequence proceeds step by step.
At the first step, a state $\rho_{s_{1}}$ is prepared from the ensemble $\mathcal{E}$ with the corresponding probability $\eta_{s_{1}}$. 
At each subsequent step $l=2,\ldots,L$, a state $\rho_{s_{l}}$ is independently prepared from the ensemble $\mathcal{E}$ with the corresponding probability $\eta_{s_{l}}$. 
Consequently, the quantum sequence
\begin{equation}\label{eq:qseq}
(\rho_{s_{1}},\ldots,\rho_{s_{L}})
\end{equation}
is prepared with the probability
\begin{equation}\label{eq:jitp}
\eta_{s_{1}}\times\cdots\times\eta_{s_{L}}
\end{equation}
for $(s_{1},\ldots,s_{L})\in\mathbb{Z}_{2}^{L}$ where $\mathbb{Z}_{2}^{L}$ denotes the Cartesian product of $L$ copies of $\{0,1\}$.

The quantum sequence in Eq.~\eqref{eq:qseq}, together with its preparation probability in Eq.~\eqref{eq:jitp}, can equivalently be represented by the tensor-producted state $\rho_{\vec{s}}$ with the joint probability $\eta_{\vec{s}}$,
\begin{equation}\label{eq:ebrb}
\eta_{\vec{s}}=\prod_{l=1}^{L}\eta_{s_{l}},~\rho_{\vec{s}}=\bigotimes_{l=1}^{L}\rho_{s_{l}},
\end{equation}
for $\vec{s}=(s_{1},\ldots,s_{L})\in\mathbb{Z}_{2}^{L}$. 
This representation gives rise to the \emph{quantum sequence ensemble}
\begin{equation}\label{eq:lfeb}
\mathcal{E}^{\otimes L}=\{\eta_{\vec{s}},\rho_{\vec{s}}\}_{\vec{s}\in\mathbb{Z}_{2}^{L}}.
\end{equation}
We note that the quantum sequences of $\mathcal{E}^{\otimes L}$ are PPT whenever the states $\rho_{0}$ and $\rho_{1}$ of the ensemble $\mathcal{E}$ are PPT, and are mutually orthogonal whenever $\rho_{0}$ and $\rho_{1}$ are orthogonal.

Now, let us consider the situation of guessing the parity $\omega_{2}(\vec{s})$ for a quantum sequence $\rho_{\vec{s}}$ prepared from the ensemble $\mathcal{E}^{\otimes L}$ in Eq.~\eqref{eq:lfeb}.
Here, $\omega_{2}(\vec{s})$ denotes the modulo-$2$ sum of $s_{1},\ldots,s_{L}$, that is,
\begin{equation}\label{eq:mtsb}
\omega_{2}(\vec{s})=\sum_{l=1}^{L}s_{l}\pmod 2,
\end{equation}
for $\vec{s}=(s_{1},\ldots,s_{L})\in\mathbb{Z}_{2}^{L}$.
This situation is equivalent to discriminating the states from the two-state ensemble 
\begin{equation}\label{eq:cgeb}
\mathcal{E}^{(L)}=\{\eta_{0}^{(L)},\rho_{0}^{(L)};\eta_{1}^{(L)},\rho_{1}^{(L)}\}
\end{equation}
where each state $\rho_{i}^{(L)}$ is prepared with probability $\eta_{i}^{(L)}$, and
\begin{equation}\label{eq:cgsp}
\eta_{i}^{(L)}=\sum_{\substack{\vec{s}\in\mathbb{Z}_{2}^{L}\\ \omega_{2}(\vec{s})=i}}\eta_{\vec{s}},~
\rho_{i}^{(L)}=\frac{1}{\eta_{i}^{(L)}}\sum_{\substack{\vec{s}\in\mathbb{Z}_{2}^{L}\\ \omega_{2}(\vec{s})=i}}\eta_{\vec{s}}\rho_{\vec{s}},
\end{equation}
for $i\in\{0,1\}$.

The following proposition provides a sufficient condition under which the optimal PPT discrimination of $\mathcal{E}^{(L)}$ in Eq.~\eqref{eq:cgeb} converges exponentially to random guessing as $L$ increases\cite{ha20252,ha20253}.
\begin{proposition}\label{pro:scpt}
For a two-party two-state ensemble $\mathcal{E}=\{\eta_{0},\rho_{0};\eta_{1},\rho_{1}\}$, if there exists $H\in\mathbb{H}$ satisfying
\begin{equation}\label{eq:idsc}
H+H^{\PT_{1}}=\Lambda_{\mathcal{E}},~
\Tr|H|+\Tr|H^{\PT_{1}}|\leqslant 1,~
\Tr|H|<\tfrac{1}{2}
\end{equation}
where $\Lambda_{\mathcal{E}}$ is defined in Eq.~\eqref{eq:lame}, then 
\begin{equation}\label{eq:smup}
p_{\sf PPT}(\mathcal{E}^{(L)})\leqslant\tfrac{1}{2}+\tfrac{1}{2}\big(4\Tr|H|\Tr|H^{\PT_{1}}|\big)^{\frac{L}{2}}
\end{equation}
for any positive integer $L$.
In this case, $p_{\sf PPT}(\mathcal{E}^{(L)})$ converges exponentially to $1/2$ as $L$ increases.
\end{proposition}

Since the states $\rho_{0}^{(L)}$ and $\rho_{1}^{(L)}$ in Eq.~\eqref{eq:cgsp} are mixtures of quantum sequences from $\mathcal{E}^{\otimes L}$ in Eq.~\eqref{eq:lfeb}, it follows that $\rho_{0}^{(L)}$ and $\rho_{1}^{(L)}$ are orthogonal PPT states whenever $\rho_{0}$ and $\rho_{1}$ are orthogonal PPT states.
Therefore, starting from any orthogonal PPT state ensemble $\mathcal{E}$ satisfying Condition~\eqref{eq:idsc} of Proposition~\ref{pro:scpt}, one can construct an orthogonal PPT state ensemble $\mathcal{E}^{(L)}$ whose optimal PPT discrimination probability approaches $\tfrac{1}{2}$ as $L$ increases.
Moreover, the above construction preserves separability: if $\rho_{0}$ and $\rho_{1}$ are separable, then $\rho_{0}^{(L)}$ and $\rho_{1}^{(L)}$ are also separable.

The following example provides an orthogonal separable state ensemble satisfying Condition~\eqref{eq:idsc} of Proposition~\ref{pro:scpt}. 
\begin{example}[\cite{ha20252}]\label{ex:osse}
Let us consider the $3\otimes 3$ orthogonal separable state ensemble $\mathcal{E}=\{\eta_{0},\rho_{0};\eta_{1},\rho_{1}\}$ consisting of
\begin{align}\label{eq:qdde}
\eta_{0}=\frac{1}{2},~
\rho_{0}=\frac{1}{4}\sum_{i=0}^{3}&\ket{\alpha_{i}}\!\bra{\alpha_{i}}\otimes\ket{\alpha_{i}}\!\bra{\alpha_{i}},\nonumber\\
\eta_{1}=\frac{1}{2},~
\rho_{1}=\frac{1}{6}\sum_{i=0}^{2}&\big(\ket{\beta_{i}^{+}}\!\bra{\beta_{i}^{+}}\otimes\ket{\beta_{i}^{-}}\!\bra{\beta_{i}^{-}}
+\ket{\beta_{i}^{-}}\!\bra{\beta_{i}^{-}}\otimes\ket{\beta_{i}^{+}}\!\bra{\beta_{i}^{+}}\big)
\end{align}
where
\begin{align}\label{eq:qddf}
\ket{\alpha_{0}}=\tfrac{1}{\sqrt{3}}(\ket{0}+\ket{1}+\ket{2}),&~
\ket{\beta_{0}^{\pm}}=\tfrac{1}{\sqrt{2}}(\ket{0}\pm\ket{1}),\nonumber\\
\ket{\alpha_{1}}=\tfrac{1}{\sqrt{3}}(\ket{0}-\ket{1}-\ket{2}),&~
\ket{\beta_{1}^{\pm}}=\tfrac{1}{\sqrt{2}}(\ket{1}\pm\ket{2}),\nonumber\\
\ket{\alpha_{2}}=\tfrac{1}{\sqrt{3}}(\ket{0}-\ket{1}+\ket{2}),&~
\ket{\beta_{2}^{\pm}}=\tfrac{1}{\sqrt{2}}(\ket{2}\pm\ket{0}),\nonumber\\
\ket{\alpha_{3}}=\tfrac{1}{\sqrt{3}}(\ket{0}+\ket{1}-\ket{2}).&~
\end{align}
Note that the states $\rho_{0}$ and $\rho_{1}$ are PPT because every separable state is a PPT state.
\end{example}

For the Hermitian operator
\begin{equation}\label{eq:heo}
H=\frac{1}{6}\ket{\Gamma_{0}}\!\bra{\Gamma_{0}}
-\frac{5}{48}\sum_{i=1}^{2}\ket{\Gamma_{i}}\!\bra{\Gamma_{i}}
+\frac{1}{72}\sum_{i=0}^{2}\ket{\Gamma_{i}^{+}}\!\bra{\Gamma_{i}^{+}}
\end{equation}
where $\{\ket{\Gamma_{i}},\ket{\Gamma_{i}^{+}},\ket{\Gamma_{i}^{-}}\}_{i=0}^{2}$ is an orthonormal basis of $\mathbb{C}^{3}\otimes\mathbb{C}^{3}$,
\begin{equation}\label{eq:xis}
\begin{array}{ll}
\ket{\Gamma_{0}}=\frac{1}{\sqrt{3}}(\ket{00}+\ket{11}+\ket{22}),
&
\ket{\Gamma_{0}^{\pm}}=\frac{1}{\sqrt{2}}(\ket{12}\pm\ket{21}),
\\[2mm]
\ket{\Gamma_{1}}=\frac{1}{\sqrt{3}}(\ket{00}+e^{\frac{\mathrm{i}2\pi}{3}}\ket{11}+e^{\frac{\mathrm{i}4\pi}{3}}\ket{22}),
&
\ket{\Gamma_{1}^{\pm}}=\frac{1}{\sqrt{2}}(\ket{20}\pm\ket{02}),
\\[2mm]
\ket{\Gamma_{2}}=\frac{1}{\sqrt{3}}(\ket{00}+e^{\frac{\mathrm{i}4\pi}{3}}\ket{11}+e^{\frac{\mathrm{i}2\pi}{3}}\ket{22}),
&
\ket{\Gamma_{2}^{\pm}}=\frac{1}{\sqrt{2}}(\ket{01}\pm\ket{10}),
\end{array}
\end{equation}
it is straightforward to check the equality condition in Eq.~\eqref{eq:idsc} and its partial transposition
\begin{equation}\label{eq:pthe}
H^{\PT}=
-\frac{1}{48}\sum_{i=1}^{2}\ket{\Gamma_{i}}\!\bra{\Gamma_{i}}
+\sum_{i=0}^{2}\Big(\frac{7}{72}\ket{\Gamma_{i}^{+}}\!\bra{\Gamma_{i}^{+}}
-\frac{1}{12}\ket{\Gamma_{i}^{-}}\!\bra{\Gamma_{i}^{-}}\Big).
\end{equation}

From Eqs.~\eqref{eq:heo} and \eqref{eq:pthe}, we can easily see that
\begin{equation}\label{eq:etnr}
\Tr|H|=\tfrac{5}{12},~\Tr|H^{\PT}|=\tfrac{7}{12},
\end{equation}
therefore Proposition~\ref{pro:scpt} together with Eq.~\eqref{eq:etnr} leads us to
\begin{equation}\label{eq:emup}
p_{\sf PPT}(\mathcal{E}^{(L)})\leqslant\tfrac{1}{2}+\tfrac{1}{2}\big(\tfrac{35}{36}\big)^{\frac{L}{2}}
\end{equation}
for any positive integer $L$.
In other words, the ensemble $\mathcal{E}$ in Example~\ref{ex:osse} can be used to construct an orthogonal separable-state ensemble whose optimal PPT discrimination is arbitrarily close to random guessing.

In particular, for any $\epsilon>0$, let $\mathcal{T}=\mathcal{E}^{(L)}$, where $L$ is chosen such that
\begin{equation}\label{eq:etrl}
\frac{1}{2}\left(\frac{35}{36}\right)^{L/2}\leqslant\epsilon,
\end{equation}
or equivalently,
\begin{equation}\label{eq:lgch}
L\geqslant\frac{2\ln(2\epsilon)}{\ln(35/36)}.
\end{equation}
Since the two states in $\mathcal{E}^{(L)}$ are orthogonal and separable, we have $p_{\sf G}(\mathcal{T})=1$, while Eq.~\eqref{eq:etrl} gives
\begin{equation}\label{eq:ptvb}
p_{\sf PPT}(\mathcal{T})\leqslant\frac{1}{2}+\frac{1}{2}\left(\frac{35}{36}\right)^{L/2}\leqslant\frac{1}{2}+\epsilon.
\end{equation}
Therefore, for every $\epsilon>0$, one can explicitly choose a finite $L$ so that $\mathcal{T}$ consists of two orthogonal separable states satisfying 
\begin{equation}\label{eq:ptau}
p_{\sf G}(\mathcal{T})=1,~p_{\sf PPT}(\mathcal{T})\leqslant\frac{1}{2}+\epsilon.
\end{equation}



\begin{thebibliography}{99}
\bibitem{dual}
When $C$ is a cone in a real vector space $V$ equipped with an inner product $\langle \cdot,\cdot \rangle$, the dual cone of $C$ is defined as the set of all vectors $v\in V$ satisfying $\langle v,w \rangle\geqslant0$ for all $w\in C$.
\bibitem{pptp}
In two-party quantum systems, the eigenvalues of a partially transposed operator do not depend on the choice of basis or on the subsystem to be transposed.
\bibitem{sht}   
When $S$ and $S'$ are disjoint convex sets in a real vector space $V$ with an inner product $\langle\cdot,\cdot\rangle$, there exist $x\in\mathbb{R}$ and $\vec{y}\in V\setminus\{\vec{0}\}$ such that $\langle\vec{v},\vec{y}\rangle\leqslant x\leqslant\langle \vec{w},\vec{y}\rangle$ for all $\vec{v}\in S$ and all $\vec{w}\in S'$.
\end{thebibliography}
\end{document}